\documentclass[12pt,onecolumn,journal]{IEEEtran}%

 \newtheorem{theorem}{Theorem}
 \newtheorem{lemma}{Lemma}
 
 \newtheorem{corollary}{Corollary}
 \newtheorem{definition}{Definition}
 \newenvironment{proof}{\begin{trivlist} \item[]{\em Proof.}}{\end{trivlist}}

\usepackage{amssymb, latexsym}

\usepackage{setspace}

\usepackage{cite}

\ifCLASSINFOpdf
  \usepackage[pdftex]{graphicx}
  \graphicspath{{../pdf/}{../jpeg/}}
  \DeclareGraphicsExtensions{.pdf,.jpeg,.png}
\else
   \usepackage[dvips]{graphicx}
   \graphicspath{{../eps/}}
   \DeclareGraphicsExtensions{.eps}
\fi
\usepackage[cmex10]{amsmath}
\usepackage{algorithm}
\usepackage{algorithmic}

\usepackage{array}
\usepackage{fixltx2e}

\usepackage{stfloats}
\begin{document}


%
\title{Ramanujan subspace pursuit for signal periodic decomposition}
%
%
%

\author{Deng~Shi-Wen*,
        Han~Ji-Qing*,~\IEEEmembership{Member,~IEEE,}
        \thanks{Manuscript received May X, 2015; revised XX XX, 2015.}

\IEEEcompsocitemizethanks{
\IEEEcompsocthanksitem *Deng~Shi-Wen is with School of Mathematical Sciences, Harbin Normal University, Harbin, China (e-mail: dengswen@gmail.com).\protect
\IEEEcompsocthanksitem *Han~Ji-Qing is with School of Computer Science and Technology, Harbin Institute of Technology, Harbin, China (e-mail: jqhan@hit.edu.cn).\protect\\
}
}

%
%

\markboth{IEEE transactions on signal processing,~Vol.~X, No.~X, XX~2015}%
{Shell \MakeLowercase{\textit{et al.}}: Bare Demo of IEEEtran.cls for Computer Society Journals}
%



\maketitle

\begin{abstract}


The period estimation and periodic decomposition of a signal are the long-standing problems in the field of signal processing and biomolecular sequence analysis.
To address such problems, we introduce the Ramanujan subspace pursuit (RSP) based on the Ramanujan subspace.
As a greedy iterative algorithm, the RSP can uniquely decompose any signal into a sum of exactly periodic components, by selecting and removing the most dominant periodic component from the residual signal in each iteration.
In the RSP, a novel periodicity metric is derived based on the energy of the exactly periodic component obtained by orthogonally projecting the residual signal into the Ramanujan subspace, and is then used to select the most dominant periodic component in each iteration.
To reduce the computational cost of the RSP,
we also propose the fast RSP (FRSP) based on the relationship between the periodic subspace and the Ramanujan subspace, and based on the maximum likelihood estimation of the energy of the periodic component in the periodic subspace.
The fast RSP has a lower computational cost and can decompose a signal of length $N$ into the sum of $K$ exactly periodic components in $ \mathcal{O}(K N\log N)$.
In addition,
our results show that the RSP outperforms the current algorithms for period estimation.

\end{abstract}



%
\IEEEpeerreviewmaketitle

\section{Introduction}

There exist many periodic patterns, such as periodicities, repetitions and regularities, in the Biomolecular sequence \cite{Anastassiou2001,Sharma2004,Arora2008,Adalbjornsson2015}, speech \cite{Roux2007}, music \cite{Sethares2005}, and machine vibration signal \cite{Muresan2003}.
To identify these periodic patterns is a long-standing problem in many application domains.
However, the traditional methods, such as the discrete Fourier transform (DFT), periodogram, and autocorrelation, cannot effectively detect these patterns, especially for the detection of the hidden periodicities \cite{Sethares99}.
Recently, various approaches have been proposed to perform the period estimation and the periodic decomposition of a signal based on subspace decomposition \cite{Sethares99, Muresan2003,Vaidyanathan2014_1,Vaidyanathan2014_2}, sparse decomposition \cite{Vaidyanathan2014C,Nakashizuka2008}, and integer period discrete Fourier transform (IPDFT) \cite{Epps2008,Epps2009}.
Unfortunately, most of these approaches usually suffer from some limitations and drawbacks, such as failure to accurately identify periodic components, limitation to be detected all periods, and higher computational cost.

By generating a set of so called `periodic subspaces', the algorithms based on subspace decomposition \cite{Sethares99,Muresan2003,Vaidyanathan2014_1,Vaidyanathan2014_2} can model and capture the periodic components of a signal by projecting the signal into these subspaces.
Each subspace consists of all the possible periodic signals (or components) with a specific period, and hence it is important for these methods to first define `what is the periodic signal?'.
With the conventional concept of the periodicity, i.e., if $x(n)$ is the signal with the period $q$, then it satisfies that $x(n)=x(n+q)$, Sethares and Staley \cite{Sethares99} constructed the periodic subspaces and proposed the periodicity transforms (PTs) \cite{Sethares99}.
The PTs can decompose a signal into a sum of the periodic components by projecting the signal into these periodic subspaces and updating residuals by removing the periodicities, iteratively.
However, there exist two drawbacks in the PTs.
The first one is that the periodic subspaces in the PTs are not orthogonal to each other, which results in that a periodic signal (or component) can lie in different subspaces.
For example, the periodic signal $\mathbf{1}$ (the vector of all ones) with the period $1$ can lie in any periodic subspace.
The second one is that the result of the PTs depends on the order in which the periodic components are extracted.
Due to these shortcomings, the PTs fails to accurately identify the periodic components in a signal.

To eliminate the ambiguity in the definition of the periodic subspace in the PTs, Muresan and Parks \cite{Muresan2003} first introduced the key concept of the exactly periodic signal.
Specifically, if $x(n)$ is an exactly periodic signal, then it satisfies that $x(n)=x(n+q)$ where $q$ is the smallest integer period.
Based on the definition of the exactly periodic signal, the authors constructed the exactly periodic subspaces (which are orthogonal to each other) by calculating the intersection of the periodic subspaces defined in the PTs \cite{Sethares99}, and proposed the orthogonal exactly periodic subspace decomposition (EPSD).
Moreover, when the aim is only to detect the periods in the signal and is not interested in its exactly periodic components, it is unnecessary to calculate the intersection of the periodic subspaces and the orthogonal projections in them.
The EPSD can perform the period estimation by directly and efficiently computing the energy of each exactly periodic component based on the maximum likelihood estimation (MLE) \cite{Wise1976}, which is a significant advantage of the EPSD.
However, the EPSD has the limitation that the exactly periodic subspaces are orthogonal to each other only when their periods are the divisors of the signal length.
For example, when the signal length is equal $12$, only the exactly periodic subspaces associated with the periods $\{1,2,3,4,6,12 \}$ are orthogonal to each other, but the exactly periodic subspaces corresponding to the periods $\{5,7,8,9,10, 11\}$ are not orthogonal.
Due to this limitation, the EPSD cannot accurately identify all periods in a signal.

More recently, Vaidyanathan \cite{Vaidyanathan2014_1,Vaidyanathan2014_2} presented the Ramanujan subspace based on the Ramanujan sums.
As an important contribution, the Ramanujan subspace is the first exactly periodic subspace that can be directly and definitely constructed.
Many interesting properties of the Ramanujan subspace, such as orthogonality and dimension, are also discussed.
As another important contribution, the relationship between the exactly periodic components and their frequencies are completely presented.
All these works are important to further explore the periodic patterns of signals.
Moreover, based on the Ramanujan subspace, Vaidyanathan also proposed the Ramanujan periodic transform (RPT) \cite{Vaidyanathan2014_2} to decompose a signal into a sum of exactly periodic components, by orthogonally projecting the signal into a set of Ramanujan subspaces which are orthogonal to each other.
Unfortunately, the RPT also suffers from the same limitation as the ESPD \cite{Muresan2003}, i.e., to construct orthogonal Ramanujan subspaces, the signal length must be a multiple of the least common multiplier of the expected periods.

Based on the framework of the sparse representation of a signal, Vaidyanathan \cite{Vaidyanathan2014C} and Nakashizuka \cite{Nakashizuka2008} proposed their methods for the periodic decomposition, respectively, with the predefined dictionary and the algorithms of sparse decomposition such as Basis Pursuit \cite{Chen1998} and LASSO \cite{Tibshirani1996}.
In \cite{Nakashizuka2008}, the dictionary is constructed based on the basis vectors defined in \cite{Sethares99}.
These basis vectors corresponding to different periods can span a set of periodic subspaces rather than the exactly periodic subspaces,
and hence the method in \cite{Nakashizuka2008} cannot accurately represent the exactly periodic components in the signal.
In \cite{Vaidyanathan2014C}, the frequency dictionary (named Farey dictionary) is constructed by partitioning the full bandwidth according to the Farey sequence.
In fact, the Farey dictionary is the frequency representation of the union of all possible Ramanujan subspaces.
For example, given the maximum period $Q=6$, the frequency bins contained in the Farey dictionary are $\left\{
\frac{0}{1},\frac{1}{6},\frac{1}{5},\frac{1}{4},\frac{1}{3},\frac{2}{5},\frac{1}{2},\frac{3}{5},
\frac{2}{3},\frac{3}{4},\frac{4}{5},\frac{5}{6},\frac{1}{1} \right\}$.
The frequency bins contained in the Ramanujan subspaces, $\mathcal{S}_1,\mathcal{S}_2,\mathcal{S}_3,\mathcal{S}_4,\mathcal{S}_5$, and $\mathcal{S}_6$, are
$\left\{\frac{1}{1} \right\}$, $\left\{\frac{1}{2} \right\}$, $\left\{\frac{1}{3},\frac{2}{3} \right\}$, $\left\{\frac{1}{4},\frac{3}{4} \right\}$, $\left\{\frac{1}{5},\frac{2}{5},\frac{3}{5},\frac{4}{5} \right\}$ and $\left\{\frac{1}{6},\frac{5}{6} \right\}$, respectively.
Obviously, the frequencies in the Farey dictionary with the maximum period $6$ are the union of that contained in all the Ramanujan subspaces with periods from $1$ to $6$ (note that the frequencies $0$ and $1$ are same).
Therefore, the sparse representation of a signal based on the Farey dictionary can reveal its exactly periodic components.
However, all the methods based on the sparse decomposition suffer from the significantly higher computational costs.
For the method in \cite{Vaidyanathan2014C}, the size of the Farey dictionary is dramatically increased with the maximum period $Q$.
For example, when $Q=512$, the number of the atoms in the Farey dictionary is $79852$, which requires a much higher computational cost to implement the sparse decomposition.

By considering the limitations of the DFT for period estimation and periodic decomposition,
Epps et al. \cite{Epps2008} presented the IPDFT by transforming a signal according to the period rather than the frequency as done in the DFT.
The IPDFT is defined by
\begin{align*}
	X(q)=\sum_{n=0}^{N-1} x(n)  \exp \left(-j\frac{2\pi n}{q} \right), q=1,2, \cdots, Q \leq N
\end{align*}
where $Q$ is the maximum period and $N$ is the length of the signal.
By combining the IPDFT and autocorrelation of a signal, the authors further proposed the autocorrelation-IPDFT method which was later used in \cite{Epps2011} for detecting periodic fragments in DNA sequence.
Although the IPDFT improves the ability of period estimation compared with the DFT, it only represents a subset of the signals with a special period rather than the whole exactly periodic subspace.
For example, Let $Q=N=6$, the frequency bins in the DFT are
$\{ 0,\frac{1}{6},\frac{2}{6},\frac{3}{6},\frac{4}{6},\frac{5}{6}\}$ and that in the IPDFT are $\{ \frac{1}{1},\frac{1}{2}, \frac{1}{3},\frac{1}{4},\frac{1}{5},\frac{1}{6}\}$.
It seems that the IPDFT can represent all the integer periods from $1$ to $6$,
but it is not the case.
The frequency bins contained in the exactly periodic subspace corresponding to the period $q=5$ are $\{\frac{1}{5},\frac{2}{5},\frac{3}{5},\frac{4}{5} \}$, but the IPDFT only represents the signals in this subspace with frequency $\{ \frac{1}{5} \}$.
Therefore, it is not possible for the IPDFT-based methods to accurately identify the periodic components in a signal.

In this paper, we derive the Ramanujan subspace pursuit (RSP) algorithm, which can overcome the drawbacks and limitations mentioned above, based on the Ramanujan subspace.
The RSP is a greedy iterative algorithm that can decompose any signal into a sum of exactly periodic components, by selecting and removing the most dominant periodic component from the residual signal in each iteration.
Unlike the PTs, the RSP is independent of the order and yields a unique decomposition of the signal.
To accurately select the most dominant periodic component in each iteration, a novel periodicity metric in the RSP algorithm is derived based on the energy of the periodic component obtained by orthogonally projecting the residual signal into the each Ramanujan subspace.
To reduce the computational cost of the RSP, we also propose the fast RSP based on the analysis of the relationship between the periodic subspace and the exactly periodic subspace and based on the MLE of the energy of the periodic component in the periodic subspace.
The fast RSP has a lower computational cost and can decompose a signal of length $N$ into the sum of $K$ exactly periodic components in $ \mathcal{O}(K N\log N)$.

The remainder of the paper is organized as follows.
In the next section, the concepts of exactly periodic signal and Ramanujan subspace are summarized.
Section \ref{SEC:OP} introduces the orthogonal projection of a signal with any length into the Ramanujan subspace.
The periodicity metric and the original RSP algorithm is proposed in Section \ref{SEC:RSP} and the fast RSP is introduced in Section \ref{SEC:FRSP}.
Moreover, the evaluations of the proposed RSP algorithm are performed in Section \ref{SEC:Evaluation}.
Finally, Section \ref{SEC:CON} concludes this paper.

The following notations are used throughout:

1) The quantity $\phi(q)$ is the \textit{Euler's totient function}.
It is equal to the number of integers $k \in [1,q]$ satisfying $(k,q)=1$.


2) $\mathcal{S}_q$ and $\mathcal{P}_q$ denote the Ramanujan subspace (or \textit{Exactly $q$-periodic}) subspace and \textit{$q$-periodic} subspace, respectively.
The $\text{dim}(\mathcal{S}_q)$ or $\text{dim}(\mathcal{P}_q)$ denotes the dimension of the subspace $\mathcal{S}_q$ or $\mathcal{P}_q$.

3) $\mathbf{x}_q$ and $\check{\mathbf{x}}_q$ denote the \textit{Exactly $q$-periodic} component and the \textit{$q$-periodic} component of a signal $\mathbf{x}$, respectively.

4) The symbols $\mathbb{Z}$ and $\mathbb{Z}_+$ stand for the set of integers and the set of nonnegative integers, respectively.

5) The uppercase bold letters ($\mathbf{A}$, $\mathbf{P}$, etc.) represent matrices.
 $\text{range}(\mathbf{P})$ denotes the range space of $\mathbf{P}$, which is spanned by the columns of $\mathbf{P}$.

6) The notation $q_i|q$ means that $q_i$ is a divisor of $q$. And $q_i \nmid q$ represents that $q_i$ is not a divisor of $q$.

\section{Exactly periodic signal and Ramanujan subspace}

This section gives the definitions of the periodic signal and the exactly periodic signal, based on which the concepts of the periodic subspace and the exactly periodic subspace are naturally defined.
The exactly periodic subspace can be achieved by constructing the Ramanujan subspace based on the Ramanujan sums.

\subsection{Periodic signal and exactly periodic signal}

The meaning of the periodicity of a signal seems to be easily understood, but there exists some ambiguity when we use the concept of the periodicity to characterize the temporal nature of the signal.
It is necessary to formally give the concept of the periodicity before the further discussion, especially the concept of the \textit{exactly periodic} signal which is first introduced in \cite{Muresan2003}.

\begin{definition}[$q$-periodic] \label{DEF:periodic}
    The signal $x(n)$ of length $N$ ($1 \leq n \leq N$) is \textit{$q$-periodic} if there exists an integer $q$ such that
    \begin{align} \label{EQ:Period}
    x(n+q)=x(n)
    \end{align}
    for all $n$ and $q$ is called the period of $x(n)$.
\end{definition}

However, the definition of the \textit{$q$-periodic} signal is ambiguous, since any \textit{$q$-periodic} signal is also \textit{$qM$-periodic} for $M \in \mathbb{Z}_+$.
For example, the periodic signal $\mathbf{1}$ (the vector of all ones) with the period $1$ can be considered as any period according to Definition \ref{DEF:periodic}.
To eliminate this ambiguity, a more definite definition about the periodicity is given as follows.

\begin{definition}[exactly $q$-periodic]\label{DEF:Eperiodic}
    The signal $x(n)$ of length $N$ ($1 \leq n \leq N$) is \textit{exactly $p$-periodic}, if and only if $q$ is the \textbf{smallest} integer in Eq. (\ref{EQ:Period}) for all $n$ and $q$ is called the \textit{exactly period} of $x(n)$.
\end{definition}

Although the concept of the \textit{exactly period} is first introduced in \cite{Muresan2003}, it has been implicitly or definitely used in \cite{Vaidyanathan2014_1,Vaidyanathan2014_2,Tenneti2015, Nakashizuka2008}.
According to the definition of the \textit{exactly $q$-periodic} signal, the periodic signal $\mathbf{1}$ with the period $1$ is just the \textit{exactly $1$-periodic} one without any ambiguity.
Obviously, the \textit{exactly period} is an important concept for characterizing the periodicity of the signal.

With the definitions about the periodicity of a signal, we can naturally define two subspaces.
One is referred to as the \textit{$q$-periodic subspace} denoted by $\mathcal{P}_q$, which is the set of all the \textit{$q$-periodic} signals.
The other is referred to as the \textit{exactly $q$-periodic subspace} denoted by $\mathcal{S}_q$, which is the set of all the \textit{exactly $q$-periodic} signals.
The framework proposed in this paper is mainly based on the \textit{exactly periodic subspace}.
Although the method for constructing the \textit{exactly periodic subspace} was proposed based on the subspace intersection in \cite{Muresan2003}, a more effective and direct way to generate the subspace is presented by constructing the Ramanujan subspace based on the Ramanujan sums.


\subsection{Ramanujan sums} \label{SUBSEC_RS}

For any fixed integer $q$, the Ramanujan sums is a sequence with period $q$ and is defined as follows
\begin{align}
	c_q(n)=\underset{(k,q)=1}{\sum_{k=1}^q} e^{j2 \pi kn/q}
\end{align}
where $(k,q)$ denotes the greatest common divisor of $k$ and $q$.
According to \cite{Vaidyanathan2014_1}, the Ramanujan sums has the following properties.

\textit{Property 1}.
If $q$ is a prime number, then
\begin{align}
	c_q(n)= \left\{
		\begin{array}{ll}
			q-1 & \text{if $n$ is multiple of $q$} \\
			-1  &  \text{otherwise}
		\end{array}
	\right.
\end{align}

\textit{Property 2}.
If $q=p^m$ for some prime $p$ and integer $m>1$ then
\begin{align}
	c_{p^m}(n) = \left \{
		\begin{array}{ll}
			0 & \text{ if $p^{m-1} \nmid n $} \\
			-p^{m-1} & \text{ if $p^{m-1}|n$ but $p^m \nmid n$} \\
			p^{m-1}(p-1) & \text{ if $p^m|n$}
		\end{array}
	\right.
\end{align}

\textit{Property 3}.
if $q_1$ and $q_2$ are coprime, i.e., $(q_1,q_2)=1$, then
\begin{align}
	c_{q_1 q_2}(n)=c_{q_1}(n)c_{q_2}(n)
\end{align}

Instead of calculating the Ramanujan sums with recursive computation in \cite{Vaidyanathan2014_1}, we give a more direct method based on the above three properties and the fundamental theorem of arithmetic \cite{Niven2008}.
The procedure of calculating the Ramanujan sums consists of the following steps.

Firstly, according to the fundamental theorem of arithmetic\cite{Niven2008},
any positive integer $q>1$ can be uniquely represented as a product of prime powers
\begin{align} \label{Eq:theo-ari}
	q=p_1^{\alpha_1}\cdots p_k^{\alpha_k}=\prod_{i=1}^k p_i^{\alpha_i}
\end{align}
where $p_1<p_2<\cdots<p_k$ are primes and $\alpha_i \in \mathbb{Z}_+$.

Secondly, the Ramanujan sums $c_{p_i^{\alpha_i}}(n)$ of the prime power $p_i^{\alpha_i}$ can be calculated by using the \textit{Property 1} and \textit{Property 2}.

Finally, the Ramanujan sums $c_q(n)$ of $q$ can be calculated from $\{ c_{p_i^{\alpha_i}}(n) \}_{i=1}^k$ by using the \textit{Property 3}.

With the above three steps, any Ramanujan sums $c_q(n)$ of the positive integer $q$ can be directly calculated.
For example, let $q=12$, which can be represented as
\begin{align*}
	q=2^2 \cdot 3
\end{align*}
The Ramanujan sums of $2^2$ and $3$ are
\begin{align*}
    \begin{array}{ll}
    c_{2^2}(12) &=2,0,-2,0,2,0,-2,0,2,0,-2,0 \\
    c_{3}(12)   &=2,-1,-1,2,-1,-1,2,-1,-1,2,-1,-1
    \end{array}
\end{align*}
Then the Ramanujan sums $c_{12}(12)$ can be obtained, that is
\begin{align*}
	c_{12}(12)=4,0,2,0,-2,0,-4,0,-2,0,2,0
\end{align*}

\subsection{Ramanujan subspace} \label{SEC:RS}

Based on the Ramanujan sums $c_q(n)$, we can directly construct the Ramanujan subspace $\mathcal{S}_q$ associated with the period $q$.
The Ramanujan subspace $\mathcal{S}_q$ is just the \textit{exactly $q$-periodic} subspace that will be used to decompose a signal according to its periodic structure in this paper.

Given the period $q \in \mathbb{Z}+$, let $c_q(n)$ be the Ramanujan sums and $\mathbf{B}_q$ be its $q \times q$ integer circulant matrix defined by
\begin{align}
        \small
	\mathbf{B}_q = \left[
		\begin{array}{cccc}
			c_q(0) & c_q(q-1)&  \cdots& c_q(1) \\
			c_q(1) & c_q(0)&  \cdots& c_q(2) \\
			c_q(2) & c_q(1)&  \cdots& c_q(3) \\
            \vdots & \vdots &  \ddots & \vdots \\
			c_q(q-2) & c_q(q-3)&  \cdots& c_q(q-1) \\
			c_q(q-1) & c_q(q-2)&  \cdots& c_q(0)
		\end{array}
	\right]
\end{align}
Thus, the Ramanujan subspace $\mathcal{S}_q \subset \mathcal{R}^q$ can be spanned by the columns of $\mathbf{B}_q$, i.e., $\mathcal{S}_q=\text{Range}(\mathbf{B}_q)$.

According to Theorem 3 and 4 in \cite{Vaidyanathan2014_1}, we know that both the column rank of $\mathbf{B}_q$ and the dimension of the subspace $\mathcal{S}_q$ are equal to $\phi(q)$, i.e., $\text{rank}(\mathbf{B}_q)=\text{dim}(\mathcal{S}_q)=\phi(q)$.
More importantly, according to Theorem 10 in \cite{Vaidyanathan2014_1}, the Ramanujan subspace $\mathcal{S}_q$ is the \textit{exactly $q$-periodic subspace}.
Let $\mathcal{S}_q^{\bot}$ denote the orthogonal complement of $\mathcal{S}_q$.
The orthogonal projection operator (or projection matrix) $\mathbf{P}_q \in \mathcal{R}^{q \times q}$ along $\mathcal{S}_q^{\bot}$ into $\mathcal{S}_q$ is given by
\begin{align}\label{EQ:ROProjector}
	\mathbf{P}_q = \frac{\mathbf{B}_q}{q}
\end{align}
which satisfies that $\mathbf{P}_q^2 = \mathbf{P}_q$ and $\mathbf{P}_q^T = \mathbf{P}_q$.
The Ramanujan subspace $\mathcal{S}_q$ is also determined by the range space of $\mathbf{P}_q$, i.e., $\mathcal{S}_q = \text{Range}(\mathbf{P}_q)$.

The Ramanujan subspace is an important concept used in this paper, based on which the whole signal space is decomposed into a set of subspaces.
Before doing this, it is required to solve the problem of `how to orthogonally project a signal of any length into the Ramanujan subspace?' in the following section.

\section{Orthogonal projection into $\mathcal{S}_q$}
\label{SEC:OP}

In this section, we focus on the problem of orthogonally projecting a signal $\mathbf{x}$ of any length $N$ into the Ramanujan subspace $\mathcal{S}_q$.
Although a method for calculating the projection is given in \cite{Vaidyanathan2014_2} when $q$ is the divisor of the length $N$, it is not suitable for the signal of any length, i.e., $q$ is not the divisor of the length $N$.
To solve the problem, we define the orthogonal projection operator into $\mathcal{S}_q$ as follows
\begin{align}
	\text{Proj}: \mathcal{R}^N \mapsto \mathcal{S}_q
\end{align}
Then the orthogonal projection $\mathbf{x}_q$ of $\mathbf{x}$ into $\mathcal{S}_q$ can be represented as $\mathbf{x}_q=\text{Proj}(\mathbf{x},\mathcal{S}_q)$.
To implement the projection operator,
we first present a method of projecting a signal of length $N$ equal to the integer multiple of the period $q$ into $\mathcal{S}_q$, and then extend the method to the general case.

\subsection{Projection of signal with length of multiple periods}

Given the signal $\mathbf{x}\in \mathcal{R}^N$, let $N$ be equal to the integer multiple of the period $q$ such that $N=qM$ where $M \in \mathbb{Z}_+$.
The Ramanujan subspace $\mathcal{\widetilde{S}}_{q} \subset \mathcal{R}^N$ can be spanned by the columns of the following matrix
\begin{align}
	\mathbf{C}_{qM} = \frac{1}{\sqrt{M}} \left[
		\begin{array}{c}
			\mathbf{P}_q \\
			\mathbf{P}_q \\
            \vdots \\
			\mathbf{P}_q
		\end{array}
	\right] \in \mathcal{R}^{qM \times q}
\end{align}
which is the $M$ repetitions of the projection matrix $\mathbf{P}_q$ of $\mathcal{S}_q$ defined in Eq. (\ref{EQ:ROProjector}),
i.e., $\mathcal{\widetilde{S}}_{q}=\text{Range}(\mathbf{C}_{qM})$.
The subspace $\mathcal{\widetilde{S}}_{q}$ has the same dimension as $\mathcal{S}_q$, i.e., $\text{dim}(\mathcal{\widetilde{S}}_{q})=\text{dim}(\mathcal{S}_{q})=\phi(q)$.
The orthogonal projection matrix $\mathbf{\widetilde{P}}_{q}$ along the orthogonal complement subspace $\mathcal{\widetilde{S}}_{q}^\perp$ into the subspace $\mathcal{\widetilde{S}}_{q}$ is given by
\begin{align}\label{EQ:PMor}
	\mathbf{\widetilde{P}}_{q} &= \mathbf{C}_{qM}(\mathbf{C}_{qM}^T\mathbf{C}_{qM})^{-1}\mathbf{C}_{qM}^T \\
    &=
        \frac{1}{M}\left [
            \begin{array}{ccc}
              \mathbf{P}_q & \cdots & \mathbf{P}_q \\
              \vdots & \ddots & \vdots \\
              \mathbf{P}_q & \cdots & \mathbf{P}_q
              \label{Eq:PP}
            \end{array}
        \right]
\end{align}
which is a $qM \times qM$ matrix and satisfies that $\mathbf{\widetilde{P}}_{q}^2 = \mathbf{\widetilde{P}}_{q}$ and $\mathbf{\widetilde{P}}_{q}^T = \mathbf{\widetilde{P}}_{q}$.
Since the subspace $\mathcal{\widetilde{S}}_q \subset \mathcal{R}^N$ is completely determined by $\mathcal{S}_q \subset \mathcal{R}^q$, it is also denoted by $\mathcal{S}_q$ to avoid a multitude of notations as done in \cite{Vaidyanathan2014_1}.
Then, the orthogonal projection $\mathbf{x}_q$ of $\mathbf{x}$ into $\mathcal{{S}}_{q}$ can be obtained by
\begin{align} \label{EQ:ProjMul}
    \mathbf{x}_q = \text{Proj}(\mathbf{x},\mathcal{S}_q) = \mathbf{\widetilde{P}}_q \mathbf{x}, \text{  } \forall \mathbf{x} \in \mathcal{R}^N
\end{align}

\subsection{Projecting signal with any length}

Instead of modifying the project matrix as done in Eq. (\ref{Eq:PP}), we give a more effective method to calculate the projection by partitioning the $\mathbf{x}$ into a set of blocks and by calculating the projection of the mean of these blocks into $\mathcal{S}_q$.
Moreover, the method of calculating mean projection of blocks can be easily applied to the signal of any length.

\begin{theorem} \label{THE:INTEST}
Let the signal $\mathbf{x} \in \mathcal{R}^N$ where $N=qM$ be partioned into $M$ blocks
\begin{align} \label{Eq:blocks}
	\mathbf{x}=\left[
		\begin{array}{c}
			\mathbf{x}^{(1)} \\
			\mathbf{x}^{(2)} \\
            \vdots  \\
			\mathbf{x}^{(M)}
		\end{array}
	\right]
\end{align}
where $\mathbf{x}^{(m)} \in \mathcal{R}^q$ is the $m$-th block, $1 \leq m \leq M$.
Let $\mathbf{s} \in \mathcal{R}^q$ be the mean of the $M$ blocks, defined by
\begin{align}
    \mathbf{s}=\frac{1}{M} \sum_{m=1}^M \mathbf{x}^{(m)}
\end{align}
Let $\mathbf{\bar{x}}_q \in \mathcal{S}_q$ be the projection of $\mathbf{s}$ into the Ramanujan subspace $\mathcal{S}_q$, which obtained by $\mathbf{\bar{x}}_q=\mathbf{P}_q \mathbf{s}$ where $\mathbf{P}_q$ is the orthogonal projection matrix of $\mathcal{S}_q$ defined in Eq. (\ref{EQ:ROProjector}).
Then, the orthogonal projection $\mathbf{x}_q$ of $\mathbf{x}$ into the Ramanujan subspace $\mathcal{S}_q$ can be obtained by repeating $M$ times of $\mathbf{\bar{x}}_q$ as follows
\begin{align} \label{EQ:ProjM}
	\mathbf{x}_q = \text{Proj}(\mathbf{x}, \mathcal{S}_q) = \left[
		\begin{array}{c}
			\mathbf{\bar{x}}_q \\
			\mathbf{\bar{x}}_q \\
            \vdots  \\
			\mathbf{\bar{x}}_q
		\end{array}
	\right]
\end{align}
\end{theorem}

\begin{proof}
The orthogonal projection $\mathbf{x}_q$ of $\mathbf{x}$ into the subspace $\mathcal{S}_{q}$ is given by
\begin{align*}
	\mathbf{x}_{q}&=\mathbf{\widetilde{P}}_{q}\mathbf{x} \\
        &=\frac{1}{M}\left[
		\begin{array}{ccc}
			\mathbf{P}_q & \cdots & \mathbf{P}_q \\
            \mathbf{P}_q & \cdots & \mathbf{P}_q \\
            \vdots       &  \ddots  & \vdots \\
			\mathbf{P}_q & \cdots & \mathbf{P}_q
		\end{array}
	\right]
	\left[
		\begin{array}{c}
			\mathbf{x}^{(1)} \\
			\mathbf{x}^{(2)} \\
            \vdots  \\
			\mathbf{x}^{(M)}
		\end{array}
	\right]  \\
    &=\frac{1}{M}\left[
		\begin{array}{c}
			\mathbf{P}_q \sum_{k=1}^M \mathbf{x}^{(k)} \\
			\mathbf{P}_q \sum_{k=1}^M \mathbf{x}^{(k)} \\
            \vdots  \\
			\mathbf{P}_q \sum_{k=1}^M \mathbf{x}^{(k)}
		\end{array}
	\right] \\
    &=\left[
		\begin{array}{c}
			\mathbf{\bar{x}}_q\\
			\mathbf{\bar{x}}_q\\
            \vdots  \\
			\mathbf{\bar{x}}_q
		\end{array}
	\right]
\end{align*}
where $\mathbf{\bar{x}}_q=\mathbf{P}_q \mathbf{s}$ and $\mathbf{s}=\frac{1}{M} \sum_{k=1}^M \mathbf{x}^{(k)}$.

\end{proof}

The theorem \ref{THE:INTEST} reveals that the orthogonal projection $\mathbf{x}_q$ of $\mathbf{x}$ into $\mathcal{S}_q$ is only depended on the mean $\mathbf{s}$ of the blocks from $\mathbf{x}$, which provides a way to calculate the orthogonal projection $\mathbf{x}_q$ of $\mathbf{x}$ of any length.
Consider the general case where the length $N$ of the signal $\mathbf{x}$ is not the integer multiple of the period $q$.
Let $M=\lceil {N} / {q} \rceil$, where $\lceil z  \rceil$ is the smallest integer greater than or equal to $z$.
The signal $\mathbf{x}$ is padded with $\tilde{N}$ zeros to obtain a new signal $\tilde{\mathbf{x}}$ of length $qM$,
where $\tilde{N}=qM-N$.
The signal $\tilde{\mathbf{x}}$ is partitioned into $M$ blocks $\{ \mathbf{x}^{(m)} \}_{m=1}^M$, where $\mathbf{x}^{(m)} \in \mathcal{R}^q$ is the $m$-th block.
The mean of these blocks can be calculated by
\begin{align}
    \mathbf{s} = \sum_{m=1}^M \mathbf{W} \mathbf{x}^{(m)}
\end{align}
where $\mathbf{W}$ is a diagonal matrix defined by
\begin{align}
\mathbf{W}=\text{diag}(\overset{q-\tilde{N}}{\overbrace{\frac{1}{M},\cdots,\frac{1}{M}}},
    \overset{\tilde{N}}{\overbrace{\frac{1}{M-1},\cdots,\frac{1}{M-1}}} )
\end{align}
Thus, the orthogonal projection $\tilde{\mathbf{x}}_q$ of the signal $\tilde{\mathbf{x}}$ into $\mathcal{S}_q$ can be calculated with theorem \ref{THE:INTEST}.
Finally, the orthogonal projection $\mathbf{x}_q$ of $\mathbf{x}$ into $\mathcal{S}_q$ can be obtained by simply preserving the first $N$ elements.

\section{Ramanujan subspace pursuit}
\label{SEC:RSP}

A signal can be decomposed into a sum of elementary building blocks so that it can be represented in a more meaningful way.
For example, the Fourier and wavelet transform can decompose a signal into a sum of frequency components; the atom decomposition can decompose a signal into a sum of time-frequency blocks or atom components.
Similarly, our purpose is to decompose a signal into a sum of exactly periodic components to obtain its {exactly periodic} representation.
The periodic decomposition can be modeled as
\begin{align} \label{EQ:PERIODMODE}
    \mathbf{x} = \sum_{k=1}^K \mathbf{x}_{q_k} + \mathbf{r}
\end{align}
where $\mathbf{x}_{q_k}=\text{Proj}(\mathbf{x},\mathcal{S}_{q_k})$ is the \textit{exactly $q_k$-periodic} component and $\mathbf{r}$ is the residual component.
Note that, if $\mathbf{r}$ is allowed for $\|\mathbf{r}\|>0$, Eq. (\ref{EQ:PERIODMODE}) is the problem of the signal approximation, while for $\|\mathbf{r}\|=0$, Eq. (\ref{EQ:PERIODMODE}) is the problem of the signal representation.
If the periodic decomposition can accurately indicate the periodic components in a signal, then it is useful for periodic estimation, periodic analysis, and other applications in signal processing, since it reveals the intrinsic periodic structure of the signal.

To obtain the unique periodic representation, some strategy or constraint is required, such as greedy strategy as done in Matching Pursuit (MP) \cite{Mallat1993} or other methods of subspace pursuit (SP).
The greedy strategy for the periodic decomposition is performed by the following procedures: finding the most dominant periodic components and removing it from the current signal to obtain the residual signal, and then repeating the same work for the residual signal at the next iteration until a stopping criterion is met.
Compared with the methods of MP or SP, the major challenge in the periodic decomposition is how to select the dominant periodic component in each iteration, due to the definition of the dominant periodic component in a signal is not clear yet.

Given the maximum expected period $Q$, a set of Ramanujan subspaces $\{\mathcal{S}_{q}\}_{q=1}^{Q}$ are generated according to the period $q$ ranging from 1 up to $Q$ and are used to capture the exactly periodic components in the signal.
Note that, since the dimension of all these subspaces, $\sum_{q=1}^Q \phi(q)$, is far greater than the length $N$ of the signal, i.e., $\sum_{q=1}^Q \phi(q) \gg N$, these Ramanujan subspaces are not orthogonal to each other (only when $N=lcm(1,\cdots,Q)$, these Ramanujan subspaces are orthogonal to each other).
However, the most dominant periodic component cannot be simply determined in terms of the value of the energy of the component captured by the subspace, as done in MP or SP.
The reason is that the dimension $\phi(q)$  of the Ramanujan subspace is directly related to the period $q \in [1, Q]$ and the higher dimensional subspace tends to capture more energy from the signal.
Therefore, it is necessary to first define a suitable metric that can be used to select the dominant periodic component at each iteration.

\subsection{Periodicity metric}

The periodicity metric of the exactly periodic component is first defined based on the autocorrelation function.
Then the relationship of the periodicity metric with the energy of the exactly periodic component and the period is revealed.

\begin{definition}[Periodicity metric] \label{DEF:periodicity_cp}
    Given a signal $\mathbf{x} \in \mathcal{R}^N$ and the maximum period $Q$, there are $Q$ exactly periodic components $\{\mathbf{x}_q = \text{Proj}(\mathbf{x},\mathcal{S}_q) \}_{q=1}^Q$ from which the most dominant periodic component will be selected.
    The periodicity metric of the exactly periodic component $\mathbf{x}_q$ is defined by
    \begin{align}\label{Eq:PTMeasure}
        P(\mathbf{x}_q,q) \triangleq \sum_{l=0}^{M-1} \varphi_{\mathbf{x}_q}(lq)
    \end{align}
    where $M=\lceil \frac{N}{q} \rceil$ and $\varphi_{\mathbf{x}_q}(\cdot)$ is the the autocorrelation function (ACF) of $\mathbf{x}_q$ is defined by
    \begin{align}\label{EQ:ACF}
        \varphi_{\mathbf{x}_q}(k)= \sum_{j=1}^{N-k} \mathbf{x}_q(j) \mathbf{x}_q(j+k)
    \end{align}
\end{definition}

The periodicity metric $P(\mathbf{x}_q,q)$ simultaneously measures the energy $\|\mathbf{x}_q\|^2=\varphi_{\mathbf{x}_q}(0)$ of the exactly periodic component $\mathbf{x}_q$ for $l=0$ and the periodic energies $\sum_{l=1}^{M-1} \varphi_{\mathbf{x}_q}(lq)$ contributed from its periodicity for $l>0$.
Although the definition of the periodicity metric of $\mathbf{x}_q$ depends on its ACF, the following theorem reveals that the periodicity metric can be calculated based on the energy $\|\mathbf{x}_q\|^2$ of $\mathbf{x}_q$, which can be obtained by projecting the signal $\mathbf{x}$ into the Ramanujan subspace $\mathcal{S}_q$.

\begin{theorem}
    The periodicity metric of $\mathbf{x}_q$ of length $N$ can be directly calculated from the energy of $\mathbf{x}_q$ as follows
    \begin{align} \label{EQ:PMxq}
      	P(\mathbf{x}_q,q)=\frac{N+q}{2q} \|\mathbf{x}_q\|^2
    \end{align}
\end{theorem}
Moreover, if $N \gg q$, then the equivalent metric can be approximated by
\begin{align} \label{EQ:PMxqa}
	P(\mathbf{x}_q,q) \approx \frac{1}{2q} \|\mathbf{{x}}_q\|^2
\end{align}

\begin{proof}

    Let $\mathbf{x}_q = \text{Proj}(\mathbf{x},\mathcal{S}_q)$ where $\mathbf{x}\in \mathcal{R}^N$,
    then $\mathbf{x}_q$ is partitioned into $M$ blocks
    \begin{align*}
        \mathbf{x}_q = \left[
            \begin{array}{c}
                \mathbf{x}_q^{(1)} \\
                \mathbf{x}_q^{(2)} \\
                \vdots  \\
                \mathbf{x}_q^{(M)}
            \end{array}
        \right]
    \end{align*}
    where $M=\lceil \frac{N}{q} \rceil$.
    The energy $\|\mathbf{x}_q\|^2$ of $\mathbf{x}_q$ can be represented as
    \begin{footnotesize}
    \begin{align*}
	\|\mathbf{x}_q\|^2 &=M \left\| \frac{1}{M} \sum_{m=1}^M \mathbf{x}_q^{(m)} \right\|^2
            = \frac{1}{M} \left\| \sum_{m=1}^M \mathbf{x}_q^{(m)} \right\|^2    \\
            &= \frac{1}{M} \left \|
                \left [
                    \begin{array}{cccccc}
                        \mathbf{x}_q(1)&+&\cdots&+&\mathbf{x}_q(1+(M-1)q) \\
                        \mathbf{x}_q(2)&+&\cdots&+&\mathbf{x}_q(2+(M-1)q) \\
                        \vdots & \vdots & \cdots & \vdots & \vdots \\
                        \mathbf{x}_q(q)&+&\cdots&+&\mathbf{x}_q(Mq)
                    \end{array}
                \right]
                \right\|^2  \\
                    &= \frac{1}{M} \left(   \|\mathbf{x}_q\|^2 + 2\sum_{l=1}^{M-1}\varphi_{\mathbf{x}_q}(lq)  \right)
    \end{align*}
    \end{footnotesize}
    where $\varphi_{\mathbf{x}_q}(\cdot)$ is the ACF defined in Eq. (\ref{EQ:ACF}).
    Thus, we can obtain that
    \begin{align*}
    	P(\mathbf{x}_q,q)= \sum_{l=0}^{M-1}\varphi_{\mathbf{x}_q}(lq) = \frac{M+1}{2} \|\mathbf{x}_q\|^2
    \end{align*}
    And, hence
    \begin{align*}
        P(\mathbf{x}_q,q)= \frac{N+q}{2q} \|\mathbf{x}_q\|^2
    \end{align*}

    Note that, $P(\mathbf{x}_q,q)$ and $P(\mathbf{x}_q,q)/N$ are completely equivalent when they are used to select the dominant periodic component from the periodic components $\{\mathbf{x}_q\}_{q=1}^Q$, where $Q$ is the maximum period.
    Hence, the periodicity metric $P(\mathbf{x}_q,q)$ is equivalent to
    \begin{align*}
        P(\mathbf{x}_q,q) = \frac{1+N/q}{2q} \|\mathbf{{x}}_q\|^2
    \end{align*}
    Moreover, if $N \gg q$, then the exactly periodicity metric can be approximated by
    \begin{align*}
    	P(\mathbf{x}_q,q) \approx \frac{1}{2q} \|\mathbf{{x}}_q\|^2
    \end{align*}

\end{proof}

\subsection{Ramanujan subspace pursuit}

The Ramanujan subspace pursuit (RSP) is iteratively performed by approximating a signal $\mathbf{x}$ with its dominant exactly periodic components and by calculating its residual.
Let $\mathbf{x}$ be a signal of length $N$, $Q$ be the maximum period, and $\Phi(Q)=\sum_{q=1}^Q \phi(q)$.
To represent the signal $\mathbf{x}$ with its exactly periodic components, it requires that $\Phi(Q)\geq N$ so that the sum of the Ramanujan subspaces $\{ \mathcal{S}_{q} \}_{q=1}^Q$ forms the complete space $\mathcal{R}^N$, i.e., $\mathcal{S}_1+\cdots+\mathcal{S}_Q=\mathcal{R}^N$.
We call $\Phi(Q)\geq N$ the condition of the periodic representation of the RSP.
Then, the RSP is performed as follows.

Let the initial residual $\mathbf{r}^0=\mathbf{x}$.
In iteration $k$, given the maximum period $Q$ and $\Gamma = \{1,\cdots,Q\}$, the residual $\mathbf{r}^{k}$ is calculated by
\begin{align} \label{EQ:RSResidue}
	\mathbf{r}^{k}=\mathbf{r}^{k-1}-\mathbf{x}_{q_k}
\end{align}
where $\mathbf{x}_{q_k}=\text{Proj}(\mathbf{r}^{k-1},\mathcal{S}_{q_k})$.
The period $q_k$ is selected so that $\mathbf{x}_{q_k}$ has the maximum value of the periodicity metric among all the periods from $1$ to $Q$, that is
\begin{align}
q_k=\underset{q \in \Gamma}{\text{argmax}} \left\{ P\left(\mathbf{x}_{q}^{(k)},q \right) \right\}
\end{align}
where $\mathbf{x}_q^{(k)}=\text{Proj}(\mathbf{r}^{k-1},\mathcal{S}_{q})$.
Since the projection $\mathbf{x}_{q_k}$ is orthogonal to $\mathbf{r}^{k-1}$ in Eq. (\ref{EQ:RSResidue}), it implies that
\begin{align}\label{EQ:RSResidueEnergy}
	\|\mathbf{r}^{k}\|^2 = \| \mathbf{r}^{k-1} \|^2-\|\mathbf{x}_{q_k}\|^2
\end{align}
Summing Eq. (\ref{EQ:RSResidue}) over $k$ between $1$ and $K$ yields
\begin{align}
	{\mathbf{x}} = \sum_{k=1}^K \mathbf{x}_{q_k} +\mathbf{r}^{K}
\end{align}
Similarly, summing Eq. (\ref{EQ:RSResidueEnergy}) over $k$ between $1$ and $K$ gives
\begin{align}
	\|\mathbf{x}\|^2 = \sum_{k=1}^K\|\mathbf{x}_{q_k}\|^2 + \|\mathbf{r}^K\|^2
\end{align}

The following theorem proves that the RSP converges exponentially and the norm of the residual in the RSP converges to zero.

\begin{theorem}
    Let $\mathbf{x}$ be a signal of length $N$ and $Q$ be the maximum period satisfying $\Phi(Q)>N$, where $\Phi(Q)= \underset{q \in \Gamma}{\sum} \phi(q)$ and $\Gamma = \{1,\cdots, Q\}$.
    The residual $\mathbf{r}^{k}$ computed by the Ramanujan subspace pursuit in the iteration $k$ satisfies
    \begin{align}
        \|\mathbf{r}^{k}\|^2 \leq (1-\mu(\Gamma))^{k}\|\mathbf{x}\|^2
    \end{align}
    where $\mu(\Gamma)$ is a constant number and satisfies $\mu(\Gamma) \in (0,1]$.
    As a consequence,
    \begin{align} \label{EQ:DEINF}
    	\mathbf{x}=\sum_{k=1}^{+\infty} \mathbf{x}_{q_k}
    \end{align}
    and
    \begin{align}\label{EQ:ENINF}
    	\|\mathbf{x}\|^2=\sum_{k=1}^{+\infty} \|\mathbf{x}_{q_k}\|^2
    \end{align}
\end{theorem}

\begin{proof}

Suppose that the residual $\mathbf{r}^k$ is already computed.
In the next iteration, the period $q_{k+1}$ is chosen with the maximum value of the periodicity metric defined in Eq. (\ref{EQ:PMxq}), which satisfies
\begin{align}
    \frac{N+q_{k+1}}{2q_{k+1}}\|\mathbf{x}_{q_{k+1}}\|^2 \geq \underset{q \in \Gamma }{\text{max}} \frac{N+q}{2q}\|  \mathbf{x}_q^{(k+1)}\|^2
\end{align}
or
\begin{align} \label{EQ:PRule}
    \|\mathbf{x}_{q_{k+1}}\|^2 \geq \underset{q \in \Gamma }{\text{max}} \frac{(N+q)q_{k+1}}{(N+q_{k+1})q} \|\mathbf{x}_q^{(k+1)}\|^2
\end{align}
where $\mathbf{x}_{q_{k+1}}=\text{Proj}(\mathbf{r}^{k},\mathcal{S}_{q_{k+1}})$ and  $\mathbf{x}_{q}^{(k+1)}=\text{Proj}(\mathbf{r}^{k},\mathcal{S}_q)$.
Then, the decomposition at the $k+1$ iteration is
\begin{align}
    \mathbf{r}^{k+1}=\mathbf{r}^k - \mathbf{x}_{q_{k+1}}
\end{align}
and
\begin{align}
    \|\mathbf{r}^{k+1}\|^2=\|\mathbf{r}^k\|^2 - \|\mathbf{x}_{q_{k+1}}\|^2
\end{align}
The rate of decay is
\begin{align} \label{EQ:RateDecay}
    \frac{\|\mathbf{r}^{k+1}\|^2}{\|\mathbf{r}^k\|^2}=1-\frac{\| \mathbf{x}_{q_{k+1}} \|^2}{\|\mathbf{r}^k\|^2}
\end{align}
According to Eq. (\ref{EQ:PRule}) and (\ref{EQ:RateDecay}), we have
\begin{align}
    \frac{\|\mathbf{r}^{k+1}\|^2}{\|\mathbf{r}^k\|^2} \leq 1 -  \underset{q \in \Gamma}{\text{max}} \frac{(N+q)q_{k+1}}{(N+q_{k+1})q}  \frac{ \|\mathbf{x}_{q}^{(k+1)}\|^2}{\|\mathbf{r}^k\|^2} \leq 1
\end{align}
Let
\begin{align}
    \eta(\mathbf{f}, \Gamma) \triangleq \underset{q \in \Gamma}{\text{max}} \frac{(N+q)q_{k+1}}{(N+q_{k+1})q}  \frac{ \| \text{Proj}( \mathbf{f}, \mathcal{S}_q)\|^2}{\|\mathbf{f}\|^2}
\end{align}
where $\mathbf{f} \in \mathcal{R}^N$,
then
\begin{align} \label{EQ:ETA}
    \frac{\|\mathbf{r}^{k+1}\|^2}{\|\mathbf{r}^k\|^2} < 1 - \eta(\mathbf{r}^k, \Gamma)
\end{align}
Note that, for any $q$ and $q_{k+1}$ taken from $\Gamma$, it satisfies that
\begin{align} \label{EQ:NQq}
  \frac{(N+q)q_{k+1}}{(N+q_{k+1})q} \geq \alpha > 0
\end{align}
where $\alpha = \frac{N/Q+1}{N+1}$,
and we have
\begin{align}
    \eta(f, \Gamma) \geq \alpha \underset{q \in \Gamma}{\text{max}}  \frac{ \| \text{Proj}( \mathbf{f}, \mathcal{S}_q)\|^2}{\|\mathbf{f}\|^2}
\end{align}
In terms of the Theorem 12.6 in \cite{Mallat2008}, it shows that
\begin{align}
	\underset{\mathbf{f} \in \mathcal{R}^N}{\text{inf}} \underset{q \in \Gamma}{\text{max}} \frac{\|\text{Proj}(\mathbf{f},\mathcal{S}_q)\|^2}{\|\mathbf{f}\|^2}>0
\end{align}
Hence, we have
\begin{align}
	\mu(\Gamma) = \underset{\mathbf{f}\in \mathcal{R}^N}{\text{inf}} \eta(\mathbf{f},\Gamma) > 0
\end{align}
Thus, the rate of decay satisfies
\begin{align} \label{EQ:MU}
    \frac{\|\mathbf{r}^{k+1}\|^2}{\|\mathbf{r}^k\|^2} < 1- \mu(\Gamma)
\end{align}

Iterating on Eq. (\ref{EQ:MU}) proves that
\begin{align}
    \|\mathbf{r}_{k}\|^2 \leq (1-\mu(\Gamma))^k \|\mathbf{x}\|^2
\end{align}
and such that
\begin{align}
   \underset{k \rightarrow 0}{\lim} \|\mathbf{r}_k\|^2 = 0
\end{align}
Hence, both the Eq. (\ref{EQ:DEINF}) and Eq. (\ref{EQ:ENINF}) are established.

\end{proof}

Since the RSP requires to calculate all the projections of the signal $\mathbf{x}$ into each Ramanujan subspace in each iteration,
the RSP suffers from higher computational cost.
For a given signal $\mathbf{x}$ of length $N$ and the maximum period $Q$, the total computational cost of the RSP with $K$ iterations is $ \mathcal{O}(KQN^2)$.
Fortunately, we will show that the computational cost of the RSP can be significantly reduced in the next section.

%
%

\section{Fast Ramanujan subspace pursuit and periodic distance}
\label{SEC:FRSP}

The high cost of the RSP is mainly attributed to the calculation of the projections of the exactly periodic components of the signal $\mathbf{x}$ into each Ramanujan subspace.
However, it is unnecessary for the RSP to calculate these projections in each iteration.
In this section, we present a fast algorithm for the RSP based on the following results.
Firstly, the relationship between the periodic subspace $\mathcal{P}_q$ and the Ramanujan subspace $\mathcal{S}_q$ reveals that the energy $\|\mathbf{x}_q\|^2$ of the orthogonal projection in $\mathcal{S}_q$ can be iteratively calculated from the energy $\|{\check{\mathbf{x}}}_q\|^2$ of the orthogonal projection in $\mathcal{P}_q$.
Secondly, the $\|{\check{\mathbf{x}}}_q\|^2$ can be easily estimated from the ACF of the signal $\mathbf{x}$ based on the MLE.
As a result, the computational cost of the fast RSP can be reduced from $ \mathcal{O}(KQN^2)$ to $\mathcal{O}(K N^2)$ and can be further reduced to $ \mathcal{O}(KN \log N)$ with the Fast Fourier Transformation (FFT) algorithm.

\subsection{Periodic subspace and exactly periodic subspace}

\begin{lemma} \label{LEMMA1}
Let $q_1, \cdots, q_K$ be all the divisors of $q$, including $1$ and $q$, and $\mathcal{S}_{q_1}, \cdots, \mathcal{S}_{q_K}$ be the Ramanujan subspaces corresponding to these divisors, then
\begin{align}
	\mathcal{S}_{q_1}, \cdots, \mathcal{S}_{q_K} \subseteq \mathcal{P}_q
\end{align}
where the ``$=$'' is hold only when $q=1$, i.e., $\mathcal{S}_{1}=\mathcal{P}_{1}$.
\end{lemma}

\begin{proof}

Let $q_k$ be the divisor of $q$ and $M_k = q/q_{k}$.
With the Definition \ref{DEF:periodic} and \ref{DEF:Eperiodic},
any \textit{exactly $q_k$-periodic} signal $\mathbf{x}_{q_k}$ must be \textit{$(q_k M_k)$-periodic} for $M_k$, i.e., $\forall \mathbf{x}_{q_k}\in \mathcal{S}_{q_k}$, then $\mathbf{x}_{q_k} \in \mathcal{P}_q$.
However, not every \textit{$q$-periodic} signal $\mathbf{x}_q$ is \textit{exactly $q$-periodic}, i.e.,
$\exists \mathbf{x}_q \in \mathcal{P}_{q}$, but $\mathbf{x}_q \notin \mathcal{S}_{q_k}$. Then, we have
\begin{align*}
\mathcal{S}_{q_k} \subset \mathcal{P}_q
\end{align*}
Hence, we have
\begin{align*}
\mathcal{S}_{q_1}, \cdots, \mathcal{S}_{q_K} \subseteq \mathcal{P}_q
\end{align*}
where the ``$=$'' is held only when $q=1$, i.e., $\mathcal{S}_{1}=\mathcal{P}_{1}$.

\end{proof}

To clearly describe the relationship between the Ramanujan subspace $\mathcal{S}_q$ and the \textit{periodic} subspace $\mathcal{P}_q$, we introduce the following Theorem.

\begin{theorem} \label{THE:ORTH}
    Let $q_i$ and $q_j$ be any two different divisors of $q$, the corresponding Ramanujan subspaces $\mathcal{S}_{q_i}$ and $\mathcal{S}_{q_j}$ are orthogonal to each other,
    that is
    \begin{align}
        \mathcal{S}_{q_i} \perp \mathcal{S}_{q_j}
    \end{align}
    where $q_i \neq q_j$.
\end{theorem}

\begin{proof}

    Let $\mathbf{\widetilde{P}}_{q_i}$ and $\mathbf{\widetilde{P}}_{q_j}$ be the orthogonal projectors of $\mathcal{S}_{q_i}$ and $\mathcal{S}_{q_j}$ defined in Eq. (\ref{EQ:PMor}), respectively, where both $\mathbf{\widetilde{P}}_{q_i}$ and $\mathbf{\widetilde{P}}_{q_j}$ are the $N \times N$ matrices.
    Suppose $\mathbf{x} \in \mathcal{S}_{q_i}$, then  $\mathbf{\widetilde{P}}_{q_i}\mathbf{x}=\mathbf{x}$.
    The orthogonal projection of $\mathbf{x}$ into $\mathcal{S}_{q_j}$ is
    \begin{align*}
        \mathbf{\widetilde{P}}_{q_j}\mathbf{x}=\mathbf{\widetilde{P}}_{q_j}\mathbf{\widetilde{P}}_{q_i}\mathbf{x}
    \end{align*}
    According to the orthogonality of the Ramanujan sequence in Eq. (15) in \cite{Vaidyanathan2014_1},
    we know that any row of $\mathbf{\widetilde{P}}_{q_i}$ and any column of $\mathbf{\widetilde{P}}_{q_j}$ are orthogonal to each other if $q_i$ and $q_j$ are the divisors of $q$ and $q_i \neq q_j$, that is
    \begin{align*}
        \mathbf{\widetilde{P}}_{q_j}\mathbf{\widetilde{P}}_{q_i}=\mathbf{0}_{N \times N}
    \end{align*}
    Consequently, we have
    $\mathbf{\widetilde{P}}_{q_j}\mathbf{x} = \mathbf{0}_{N \times 1} $.
    Similarly, we can also obtain that $\mathbf{\widetilde{P}}_{q_i}\mathbf{x} = \mathbf{0}_{N \times 1}$, $\forall \mathbf{x} \in \mathcal{S}_{q_j}$.
    Thus, both Ramanujan subspaces $\mathcal{S}_{q_i}$ and $\mathcal{S}_{q_j}$ are orthogonal to each other
    \begin{align*}
        \mathcal{S}_{q_i} \bot \mathcal{S}_{q_j}
    \end{align*}

\end{proof}

According to Lemma \ref{LEMMA1} and Theorem \ref{THE:ORTH},
we can obtain the relationship between the periodic subspace and the Ramanujan subspace, which is introduced in the following corollary.

\begin{corollary} \label{COR:DEC}
    Let $q_1,\cdots,q_K$ be all the divisors of $q$, including $1$ and $q$, then the corresponding Ramanujan subspaces $\mathcal{S}_{q_1}, \cdots,\mathcal{S}_{q_K} \subset \mathcal{P}_q$ are orthogonal to each other.
    Moreover, these Ramanujan subspaces form an orthogonal decomposition of the \textit{$q$-periodic} space $\mathcal{P}_q$,
    that is,
    \begin{align}
        \mathcal{P}_q = \mathcal{S}_{q_1} \oplus  \cdots \oplus \mathcal{S}_{q_K}
    \end{align}
    and the dimensions of these subspaces satisfy that
    \begin{align}
        \text{dim}(\mathcal{P}_q) = \sum_{k=1}^{K}\text{dim}(\mathcal{S}_{q_k})=\sum_{k=1}^{K} \phi(q_k)=q
    \end{align}

\end{corollary}

The Corollary \ref{COR:DEC} provides a way to orthogonally decompose a \textit{$q$-periodic} signal into a sum of its {exactly periodic} components.
Let ${\mathbf{x}} \in \mathcal{P}_q$ denote the \textit{$q$-periodic} signal and $q_1,\cdots,q_K$ be all the possible divisors of $q$ (including $1$ and $q$).
With Corollary \ref{COR:DEC}, ${\mathbf{x}}$ can be orthogonally decomposed into a sum of the {exactly periodic} components $\{ \mathbf{x}_{q_k} \}_{k=1}^{K}$ which satisfy that
\begin{align}\label{EQ:SUBDEC}
    {{\mathbf{x}}} = \sum_{k=1}^{K} \mathbf{x}_{q_k}
\end{align}
and
\begin{align}\label{EQ:SUBEN}
    \|{{\mathbf{x}}}\|^2 = \sum_{k=1}^K \|\mathbf{x}_{q_k}\|^2
\end{align}

The Corollary \ref{COR:DEC} indicates the relationship between the energies of the projections in the Ramanujan subspaces and the energy of the projection in the \textit{$q$-periodic} subspace in Eq. (\ref{EQ:SUBEN}).
The former are directly related to the calculation of the periodicity metrics of these exactly periodic components in the corresponding Ramanujan subspaces.
In fact, we will show that these periodicity metrics can be calculated without projecting the signal into each Ramanujan subspace.

\subsection{Fast Ramanujan subspace pursuit}

\begin{figure*}[h]
        \begin{minipage}[b]{1.0\linewidth}
          \centering
          \centerline{\includegraphics[width=13cm]{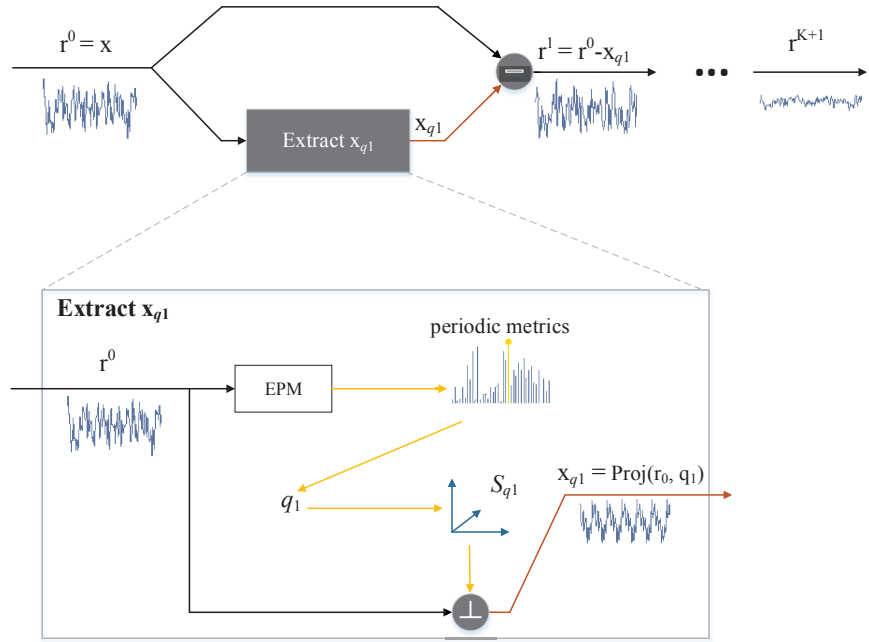}}
        \end{minipage}
        \caption{Demonstration of the fast Ramanujan subspace pursuit algorithm.}
        \label{FIG:demo_alg2}
\end{figure*}

\begin{algorithm}[b]
    \caption{\small Estimating the periodic metrics (EPM)}

    \textbf{Input:} signal $\mathbf{x}\in \mathcal{R}^N$ \\
    \textbf{Output:} periodicity metrics $P(\mathbf{x}_1,1),\cdots, P(\mathbf{x}_Q,Q)$

    \textbf{Initialize:}
    Compute $\|\mathbf{x}_1\|^2 = \|\check{\mathbf{x}}_1\|^2$ with Eq. (\ref{EQ:EMLE})

    for $q=2$ to $Q$

     \ \ \ \   Estimate $\|\check{\mathbf{x}}_q\|^2$ with Eq. (\ref{EQ:EMLE})

     \ \ \ \   Calculate $\|\mathbf{x}_q\|^2$ with Eq. (\ref{EQ:ENEPSD})

     \ \ \ \   Calculate periodicity metric $P(\mathbf{x}_q, q)$ with Eq. (\ref{EQ:PMxq}) or (\ref{EQ:PMxqa})

    end

    \label{Alg_EPSD}

\end{algorithm}

Given a signal $\mathbf{x} \in \mathcal{R}^N$, our purpose is to calculate the energy $\|\mathbf{x}_q\|^2$ of the \textit{exactly $q$-periodic} component in the Ramanujan subspace $\mathcal{S}_q$ without projecting $\mathbf{x}$ into $\mathcal{S}_q$.
Suppose that the energy $\|\check{\mathbf{x}}_q\|^2$ of the \textit{$q$-peroidic} component and the energies $\{\|\mathbf{x}_{q_k}\|^2\}_{k=1}^{K-1}$ of the exactly periodic components are already computed, where $q_1, \cdots, q_{K-1}$ are the first $K-1$ possible divisors (except $q$) of $q$.
Then, according to Eq. (\ref{EQ:SUBEN}), the energy $\| \mathbf{x}_q \|^2$ of the \textit{exactly $q$-periodic} component can be calculated by
\begin{align}\label{EQ:ENEPSD}
	\| \mathbf{x}_q \|^2 = \|\check{\mathbf{x}}_q\|^2-\sum_{k=1}^{K-1} \|\mathbf{x}_{q_k}\|^2
\end{align}
which can be used to calculate the periodicity metric with Eq. (\ref{EQ:PMxq}) or (\ref{EQ:PMxqa}).

Although the energy $\|\check{\mathbf{x}}_q\|^2$ of the \textit{$q$-peroidic} component of $\mathbf{x}$ can be obtained by projecting $\mathbf{x}$ into $\mathcal{P}_q$ with the periodic transform \cite{Sethares99},
it can be estimated by using the MLE \cite{Wise1976,Muresan2003}, that is
\begin{align}\label{EQ:EMLE}
	\|\check{\mathbf{x}}_q\|^2=\frac{q}{N}\left( \varphi_{\mathbf{x}}(0)+2\sum_{l=1}^{M-1} \varphi_{\mathbf{x}}(lq) \right)
\end{align}
where $\varphi_{\mathbf{x}}(\cdot)$ is the ACF of $\mathbf{x}$ defined in Eq. (\ref{EQ:ACF}) and $M=\lfloor N/q\rfloor$.
Thus, the energy $\|\mathbf{x}_q\|^2$ of the \textit{exactly $q$-periodic} component of $\mathbf{x}$ can be iteratively computed according to Eq. (\ref{EQ:ENEPSD}) with the initial condition $\mathbf{x}_1 = \check{\mathbf{x}}_1$.
Base on $\|\mathbf{x}_q\|^2$, the periodicity metric of the exactly periodic component $\mathbf{x}_q$ can be computed with Eq. (\ref{EQ:PMxq}) or (\ref{EQ:PMxqa}) for $q$ from $1$ up to $Q$.
The implementation of estimating the periodic metrics (EPM) is shown in Algorithm \ref{Alg_EPSD}.
Based on the Algorithm \ref{Alg_EPSD}, the fast RSP (FRSP) can be performed as shown in \ref{Alg_EstimatHP}.
Moreover, the whole FRSP algorithm is demonstrated in Fig. \ref{FIG:demo_alg2}.

Since the periodicity metrics of the exactly periodic components can be directly estimated from the ACF of the residual in each iteration, the computational cost of the FRSP can be reduced to $\mathcal{O}(K N^2)$ compared with the original RSP algorithm.
Moreover, when the ACF of the residual is more effective calculated by using the FFT in each iteration, the computational cost of the FRSP can be further reduced to $ \mathcal{O}(KN \log N)$.

\begin{algorithm}[t]
    \caption{\small Fast Ramanujan subspace pursuit (FRSP)}

    \textbf{Input:} signal $\mathbf{x}\in \mathcal{R}^N$ \\
    \textbf{Output:} Exactly periodic components $\{ \mathbf{x}_q \}_{q=1}^Q$

    \textbf{Initialize:} Set $\mathbf{r}_0=\mathbf{x}$ \\
    for $k=1,\cdots,K$

      \ \ \ \ Calculate exactly periodic measures $\{ P(\mathbf{r}^k,q) \}_{q=1}^Q$ with Algorithm \ref{Alg_EPSD}

      \ \ \ \ Select dominant period $q_k = \underset{q \in \Gamma }{\text{argmax}} \left\{ P(\mathbf{r}^k, q) \right\}$

      \ \ \ \ Generate Ramanujan subspace $S_{q_k}$

      \ \ \ \ Calculate $ \mathbf{x}_{q_k} = \text{Proj}(\mathbf{r}^k, \mathcal{S}_{q_k})$

      \ \ \ \ Update residual $\mathbf{r}^{k+1}=\mathbf{r}^k-\mathbf{x}_{q_k}$

    end 

    \label{Alg_EstimatHP}

\end{algorithm}

\subsection{Periodic energy spectrum and periodic distance}

With the FRSP algorithm, any signal $\mathbf{x}$ can be decomposed into a sum of the exactly periodic components.
The decomposition can provide a visual representation named as the periodic energy spectrum (PES) to represent the periodic structures of the signal.
The PES of the signal $\mathbf{x}$ is defined by
\begin{align} \label{EQ:PES}
    PES(q)=\|\mathbf{x}_q\|^2, \text{ for } q=1,\cdots, Q
\end{align}
where $\mathbf{x}_q$ is the \textit{exactly $q$-periodic} component obtained by using the FRSP.
Unlike the power spectrum which describes the energy of a signal is distributed over the continuous frequency in the field of signal processing,
the PES describes the energy of the signal is distributed over the discrete and integer period.
For example, Fig. \ref{FIG:SimilarityofTwoSig} (c) and (d) show the PESs of the signal $\mathbf{x}$ of length $200$ in Fig. \ref{FIG:SimilarityofTwoSig} (a) and the signal $\mathbf{y}$ of length $241$ in Fig. \ref{FIG:SimilarityofTwoSig} (b), respectively.
The PESs of the two signals reveal their similar periodic structures that both of them contain the exactly periodic components with periods: $7$, $15$, $24$, and $27$.

\begin{figure}[h]
        \begin{minipage}[h]{1.0\linewidth}
          \centering
          \centerline{\includegraphics[width=9.5cm]{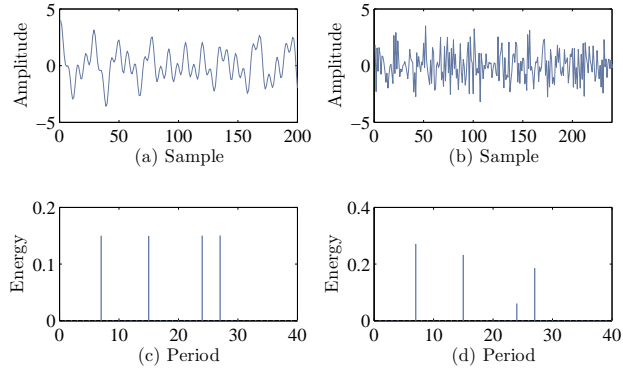}}
        \end{minipage}
        \caption{(a) Signal $\mathbf{x}$ of length $200$; (b) Signal $\mathbf{y}$ of length $241$; (c) Periodic energy spectrum of the signal $\mathbf{x}$; (d) Periodic energy spectrum of the signal $\mathbf{y}$.}
        \label{FIG:SimilarityofTwoSig}
\end{figure}

The PESs of signals provide a way to measure their distance and similarity based on their periodic structures.
Although the concept of the periodic distance is first introduced in \cite{Vlachos2005}, we will redefine it in a more general form based on the Hellinger distance (or Hellinger divergence) \cite{Jebara2004, Tran2011},
which can be seen as a symmetric approximation to the Kullback-Leibler (KL) divergence \cite{Jaakkola1999}.
The Hellinger distance can be used to measure the distance of two histograms $\mathbf{h}_A$ and $\mathbf{h}_B$
\begin{align}
D(\mathbf{h}_A,\mathbf{h}_B) = \frac{1}{2}\sum_{n=1}^N \left( \sqrt{\mathbf{h}_A(n)} - \sqrt{\mathbf{h}_B(n)}  \right)^2
\end{align}
where $\mathbf{h}_A=[\mathbf{h}_A(1),\cdots,\mathbf{h}_A(N)]$ and $\mathbf{h}_B=[\mathbf{h}_B(1),\cdots,\mathbf{h}_B(N)]$ are two histograms.

Similar with the PES in Eq. (\ref{EQ:PES}), the histogram of the energy distribution of the periodic components in the signal $\mathbf{x}$, named periodic energy histogram, can be defined by
\begin{align}
    h_E(\mathbf{x})=\frac{1}{\mathbf{\|x\|^2}} \left[\|\mathbf{x}_1\|^2, \cdots, \|\mathbf{x}_Q\|^2 \right]
\end{align}
where $Q$ is the maximum period and $\{\mathbf{x}_q\}_{q=1}^{Q}$ are the \textit{exactly $q$-periodic} components.
Note that only $K$ elements in $h_E(\mathbf{x})$ are non-zero and are obtained by using the FRSP with $K$ iterations.
Then, the Hellinger distance between $\mathbf{x}$ and $\mathbf{y}$ (note that they may have different length) is defined by
\begin{align}
D \left(h_E(\mathbf{x}),h_E(\mathbf{y}) \right)=\frac{1}{2}\sum_{q=1}^Q \left( \frac{\|\mathbf{x}_q \|}{\|\mathbf{x}\|}-\frac{\|\mathbf{y}_q\|}{\|\mathbf{y}\|}  \right)^2
\end{align}
where $\mathbf{x}_q$ and $\mathbf{y}_q$ are the \textit{exactly $q$-periodic} components of $\mathbf{x}$ and $\mathbf{y}$, respectively.
For simplicity, the periodic energy histogram of $\mathbf{x}$ is defined by
\begin{align} \label{EQ:PEHIST}
\mathbf{h_x}=\frac{1} {\mathbf{\|x\|}} \left[ \|\mathbf{x}_1\|, \cdots, \|\mathbf{x}_Q\| \right]^T
\end{align}
Then, the Hellinger distance between $\mathbf{x}$ and $\mathbf{y}$ is defined by
\begin{align}
    D(\mathbf{x}, \mathbf{y}) = \frac{1}{2}\| \mathbf{h_x}- \mathbf{h_y} \|^2
\end{align}
where $\mathbf{h_x}$ and $\mathbf{h_y}$ are the periodic energy histograms of $\mathbf{x}$ and $\mathbf{y}$ defined by Eq. (\ref{EQ:PEHIST}), respectively.
And, the cosine Hellinger distance between $\mathbf{x}$ and $\mathbf{y}$ is defined by
\begin{align}
	D_{cos}(\mathbf{x}, \mathbf{y})=\frac{1}{2}\frac{\| \mathbf{h_x}-\mathbf{h_y} \|^2}{\|\mathbf{h_x}\| \|\mathbf{h_y}\|}
\end{align}
where $D_{cos}(\mathbf{x}, \mathbf{y}) \in [0,1]$.
By the cosine Hellinger distance, we can define the periodic similarity between $\mathbf{x}$ and $\mathbf{y}$ as follows
\begin{align}\label{EQ:SIMPER}
	S(\mathbf{x},\mathbf{y}) = 1- D_{cos}(\mathbf{x},\mathbf{y})
\end{align}
where $S(\mathbf{x}, \mathbf{y}) \in [0,1]$.

With the periodic similarity defined in Eq. (\ref{EQ:SIMPER}), we can compare the similarity of two signals based on the periodicity.
For example, the signals $\mathbf{x}$ and $\mathbf{y}$, as shown in Fig. \ref{FIG:SimilarityofTwoSig} (a) and (b), are very different in their waveforms.
However, their periodic similarity $S(\mathbf{x},\mathbf{y})$ is $0.9703$ and characterizes them as very similar, because they have the similar periodic structures, i.e., both signals contain the same periodic components with periods: $7,15,24$ and $27$, as shown in Fig. \ref{FIG:SimilarityofTwoSig} (c) and (d).

\section{Evaluation}
\label{SEC:Evaluation}

To evaluate the effectiveness of the FRSP algorithm in identifying the periods, the proposed algorithm is applied to a set of synthetic signals and is compared with several approaches: Ramanujan periodicity transforms (RPT) \cite{Vaidyanathan2014_1, Vaidyanathan2014_2}, Exactly Periodic Subspace Decomposition (EPSD) \cite{Muresan2003}, and four periodicity transforms (PTs) including \textit{Mbest, BestCorrlation, BestFrequence, Small2Large} proposed in \cite{Sethares99}.
For all these algorithms, we first show their performances of identifying the periods in a synthetic signal, then evaluate their capabilities of detecting periods of signals with the different signal length and the noise level.

\subsection{Illustration of identifying periods}

The synthetic periodic signal is given by
\begin{align} \label{EQ:3PeriodSig}
x[n]=\cos \frac{2\pi n}{q_1} +\cos \frac{2\pi n}{q_2} +\cos \frac{2\pi n}{q_3}
\end{align}
which consists of three cosine signals with the periods $q_1=17, q_2=36$, and $q_3=45$.
Let the length $N$ of the signal be equal to $3060$, which is the least common multiplier of the periods, as shown in Fig. \ref{FIG:MultiPeriodsN}.
The range of the detected periods is from $1$ up to $60$ for these algorithms.
Fig. \ref{FIG:IntegerNPeriods} (a)-(g) show the PESs obtained by using these methods.
Since the signal in Eq. (\ref{EQ:3PeriodSig}) consists of integer periods for each periodic components, the FRSP, RPT EPSD, BestFrequency, and Small2Large can exactly identify three periodic components.
Although there exists a little distortion in the energies of the exactly periodic components for the \textit{BestCorrelation} in Fig. \ref{FIG:IntegerNPeriods} (e), it can also exactly identify the three periods.
However, the \textit{Mbest} in Fig. \ref{FIG:IntegerNPeriods} (d) only correctly identifies the periodic component at period $45$.

\setcounter{figure}{2}
\begin{figure}[b]
        \begin{minipage}[h]{1.0\linewidth}
          \centering
          \centerline{\includegraphics[width=9cm]{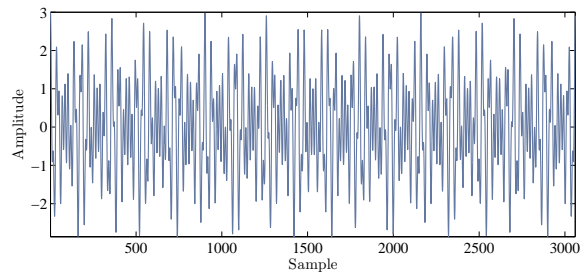}}
        \end{minipage}
        \caption{Periodic signal of length $3060$ with three periods: $17$, $36$, and $45$.}
        \label{FIG:MultiPeriodsN}
\end{figure}

\setcounter{figure}{3}
\begin{figure*}[h]
        \begin{minipage}[htb]{1.0\linewidth}
          \centering
          \centerline{\includegraphics[width=13cm]{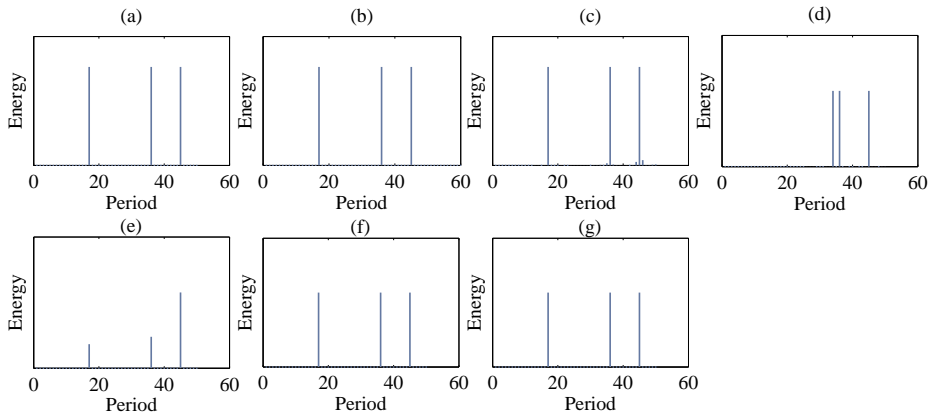}}
        \end{minipage}
        \caption{Periodic energy spectrums of the signal in Fig. (\ref{FIG:MultiPeriodsN}) by using several methods: (a) FRSP; (b) RPT; (c) EPSD; (d) Mbest; (e) BestCorrelation; (f) BestFrequency; (g) Small2Large.}
        \label{FIG:IntegerNPeriods}
\end{figure*}

Next, consider the performance of identifying the hidden periods of the synthetic periodic signal of length $511$ as shown in Fig. \ref{FIG:MultiPeriodsNI}.
Since the length of the signal is not the multiplier of any hidden periods, the signal consists of non-integer periods for three periodic components of periods: $17,36$, and $45$.
Fig. \ref{FIG:NonIntegerNPeriods} shows the PESs obtained by using these algorithms.
The FRSP in Fig. \ref{FIG:NonIntegerNPeriods} (a) can exactly identify three hidden periods when the length of the signal is reduced from $3060$ to $511$.
However, the RPT in Fig. \ref{FIG:NonIntegerNPeriods} (b) is completely failed since none of the hidden periods $17,36,45$ is the divisors of the signal length $N=511$.
In fact, two false periods $1$ and $511$ (the second period peak is not shown in this plot) are detected by the RPT.
Fig. \ref{FIG:NonIntegerNPeriods} (c) show that some false periods are obviously detected by the ESPD due to the same reason as the RPT.
Although several false periods are detected in Fig. \ref{FIG:NonIntegerNPeriods} (e), the \textit{BestCorrlation} achieves the similar result when the length of the signal is changed from $3060$ to $511$.
Both the \textit{BestFrequency} in Fig. \ref{FIG:NonIntegerNPeriods} (f) and \textit{Small2Large} in Fig.  \ref{FIG:NonIntegerNPeriods} (g) can correctly identify two periods $36$ and $45$, in spite of missing the period $17$.
Although several false periods are detected,
the \textit{Mbest} in Fig. \ref{FIG:NonIntegerNPeriods} (d) identifies all three periods and
obtains a better result when the length of the signal is changed.

\subsection{Robustness of identifying hidden periods}

Instead of using the signal in Eq. (\ref{EQ:3PeriodSig}) with fixed periods and length, we generate the test signal in a more general way to investigate the robustness of these algorithms for identifying hidden periods.
The test signal is mixed by four periodic components, each of which is generated as follows:
1) Randomly selecting a period $q \in [1,100]$;
2) Generating a random component of length $q$ and repeating it multiple times.

To evaluate the performances of these algorithms for identifying hidden periods of the test signal of different length, the signal length $N$ is varied from $100$ to $1000$.
For each algorithm, the range of the detected period is from $1$ up to $120$, and the maximum of the iteration is $10$.
For each signal length $N$, the detection process is repeated $50$ times with randomly selecting four periods and random components.
Than, the periodic similarities of these algorithms are calculated by using Eq. (\ref{EQ:SIMPER}) between the original test signal and detection results of these algorithms.
Since the RPT can correctly detect any periods which are not the divisors of the length of the tested signal, the RPT is not used for this test.
Fig. \ref{FIG:DetectPeriodsinN_with_std} shows the plot of the average periodic similarities of these algorithms as functions of the signal length $N$.
According to the results, the FRSP achieves the best performance among all these algorithms.
Especially, it shows that even if the length of the tested signal is $200$, which may only contain two periods of the periodic component with the hidden period $100$, the FRSP also can obtain nearly $0.8$ of the periodic similarity.
Note that, the \textit{BestCorrelation} and \textit{Mbest} achieve better performance than EPSD, \textit{BestFrequency} and \textit{Small2Large}.

\setcounter{figure}{4}
\begin{figure}[h]
        \begin{minipage}[h]{1.0\linewidth}
          \centering
          \centerline{\includegraphics[width=9cm]{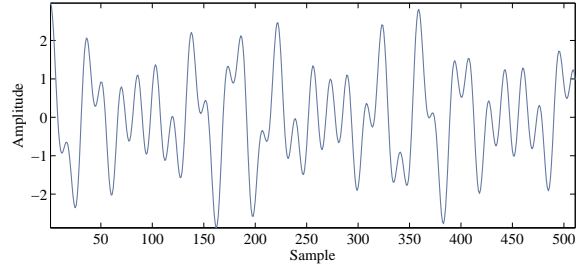}}
        \end{minipage}
        \caption{Periodic signal of length $511$ with three periods: $17$, $36$, and $45$.}
        \label{FIG:MultiPeriodsNI}
\end{figure}

\setcounter{figure}{5}
\begin{figure*}[h]
        \begin{minipage}[b]{1.0\linewidth}
          \centering
          \centerline{\includegraphics[width=13cm]{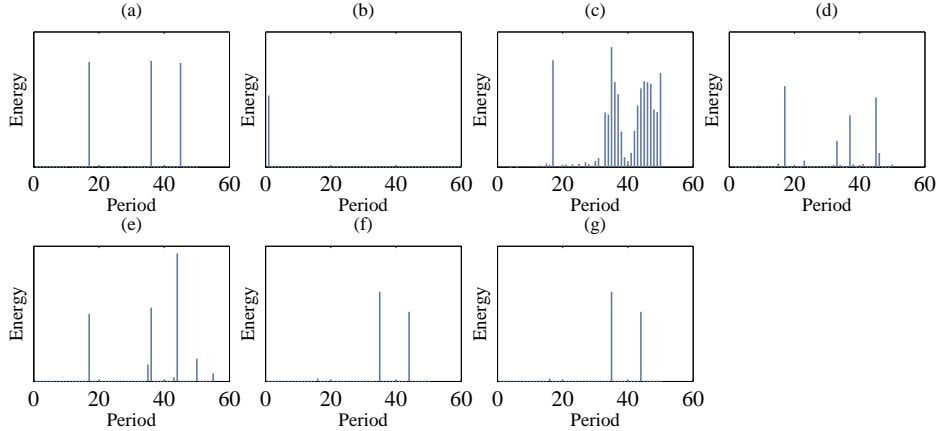}}
        \end{minipage}
        \caption{Periodic energy spectrums of the signal in Fig. (\ref{FIG:MultiPeriodsNI}) by using several methods: (a) FRSP; (b) RPT; (c) EPSD; (d) Mbest; (e) BestCorrelation; (f) BestFrequency; (g) Small2Large.}
        \label{FIG:NonIntegerNPeriods}
\end{figure*}

Similarly, the performances of these algorithms for identifying hidden periods in noise are evaluated by adding the white noise to the tested signal of fixed length $500$.
The white noise, which comes from NOISEX-92 database \cite{NOISEX92}, is added to the tested signal at the signal-to-noise ratio (SNR) varied from $-20$ to $40$ dB.
For each algorithm, the range of the detected period is from $1$ up to $120$, and the maximum of the iteration is $10$.
For each noise level (SNR), the detection process is repeated $50$ times with randomly selecting four periods and random components.
Fig. \ref{FIG:DetectPeriodsinNoise} shows the plot of the average periodic similarities of these algorithms as functions of SNR of the tested signal.
According to the results, the FRSP achieves the best performance among all these algorithms.
Especially, it shows that even when the $\text{SNR}$ is equal to $0$ dB, the FRSP also can obtain nearly $0.8$ of the periodic similarity.
Note that, the \textit{BestCorrelation} and \textit{Mbest} achieve better performance than EPSD, \textit{BestFrequency} and \textit{Small2Large}.

\section{Conclusion}
\label{SEC:CON}

In this paper, we presented a novel algorithm named Ramanujan subspace pursuit (RSP) for period estimation and periodic decomposition of a signal based on the Ramanujan subspace.
The RSP is a greedy iterative algorithm and is implemented by selecting and removing the most dominant periodic component from the residual signal in the current iteration.
The RSP can identify any periodic components corresponding to the periods from $1$ the maximum period $Q$ and hence successfully overcomes the limitation in ESPD and RPT, which only the periods being the divisors of the length of the signal can be accurately detected.
In addition to the period estimation, the RSP also provides a way to decompose the signal into the sum of the exactly periodic components, just like the method of the atomic decomposition used in signal processing.
Moreover, the fast RSP, based on the MLE of the energy of the periodic component in the periodic subspace, has a lower computational cost and can decompose a signal of length $N$ in $ \mathcal{O}(K N\log N)$.
Our results show that the RSP outperforms the current algorithms such as RPT, ESPD and PTs.

\begin{figure}[h]
        \begin{minipage}[htb]{1.0\linewidth}
          \centering
          \centerline{\includegraphics[width=9cm]{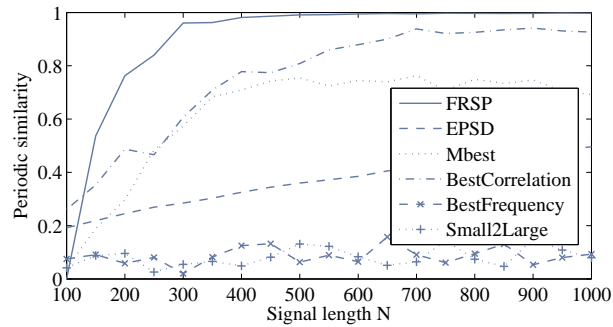}}
        \end{minipage}
        \caption{Average periodic similarity v.s. signal length $N$.}
        \label{FIG:DetectPeriodsinN_with_std}
\end{figure}

\begin{figure}[h]
        \begin{minipage}[b]{1.0\linewidth}
          \centering
          \centerline{\includegraphics[width=9cm]{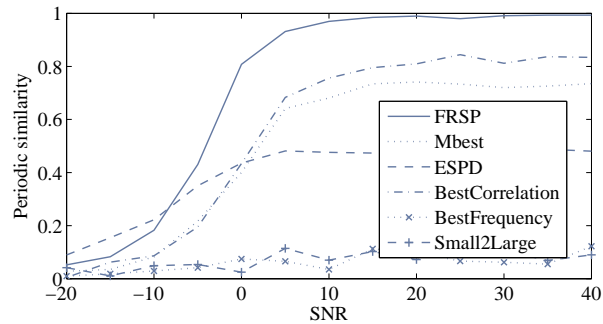}}
        \end{minipage}
        \caption{Average periodic similarity v.s. $SNR$.}
        \label{FIG:DetectPeriodsinNoise}
\end{figure}

\section* {Appendix}

\ifCLASSOPTIONcompsoc
  \section*{Acknowledgments}
\else
  \section*{Acknowledgment}
\fi
 This work was supported in part by the Major Research plan of the National Natural Science Foundation of China (No. 91120303), National Natural Science Foundation of China (No. 91220301), Natural Science Foundation of Heilongjiang Province of China (No. F2015012), and Academic Core Funding of Young Projects of Harbin Normal University of China (No. KGB201225).

\ifCLASSOPTIONcaptionsoff
  \newpage
\fi



%

%

\begin{IEEEbiography}{Deng Shi-Wen}
Shi-Wen Deng received the B.E degree in the Institute of Technology from Jia Mu Si University, JiaMuSi, China, in 1997, the M.E in The school of Computer Science from Harbin Normal University, Harbin, China, in 2005, and the Ph.D degree in the school of Computer Science from Harbin Institute of Technology in 2012. Currently, he is with the School of Mathematical Sciences, Harbin Normal University, Harbin, China. His research interests are in the area of speech and audio signal processing, including content-based audio analysis, noise suppression, speech/audio classification/detection.
\end{IEEEbiography}

\begin{IEEEbiographynophoto}{Han Ji-Qing}
Ji-Qing Han received the B.S., M.S. in electrical engineering, and Ph.D. degrees in computer science from the Harbin Institute of Technology, Harbin, China, in 1987, 1990, and 1998, respectively. Currently, he is the associate dean of the school of Computer Science and Technology, Harbin Institute of Technology. He is a member of IEEE, member of the editorial board of Journal of Chinese Information Processing, and member of the editorial board of the Journal of Data Acquisition and Processing. Prof. Han is undertaking several projects from the National Natural Science Foundation, 863Hi-tech Program, National Basic Research Program. He has won three Second Prize and two Third Prize awards of Science and Technology of Ministry/Province. He has published more than 100 papers and 2 books. His research fields include speech signal processing and audio information processing.
\end{IEEEbiographynophoto}





\end{document}